\documentclass{article}
\usepackage{fullpage}
\usepackage{times}
\usepackage{amsmath,amsfonts,amssymb,amsthm}
\usepackage{framed}

\begin{document}

\newtheorem{restate}{Theorem}
\newtheorem{theorem}{Theorem}
\newtheorem{proposition}{Proposition}
\newtheorem{corollary}{Corollary}
\newtheorem{lemma}{Lemma}
\newtheorem{claim}{Claim}
\newtheorem{definition}{Definition}

\title{An Improved Combinatorial Algorithm for Boolean Matrix Multiplication}
\author{Huacheng Yu\thanks{Supported in part by NSF CCF-1212372.} \\ \vspace{-10pt} \\ Stanford University}

\maketitle

\begin{abstract}
We present a new combinatorial algorithm for triangle finding and Boolean matrix multiplication that runs in $\hat{O}(n^3/\log^4 n)$ time, where the $\hat{O}$ notation suppresses poly(loglog) factors. 
This improves the previous best combinatorial algorithm by Chan~\cite{Chan15} that runs in $\hat{O}(n^3/\log^3 n)$ time. Our algorithm generalizes the divide-and-conquer strategy of Chan's algorithm.

Moreover, we propose a general framework for detecting triangles in graphs and computing Boolean matrix multiplication. Roughly speaking, if we can find the ``easy parts'' of a given instance efficiently, we can solve the whole problem faster than $n^3$. 
\end{abstract}

\section{Introduction}
Boolean matrix multiplication (BMM) is one of the most fundamental problems in computer science.
It has many applications to triangle finding, transitive closure, context-free grammar parsing, etc~\cite{Furman70},~\cite{Munro71},~\cite{FM71},~\cite{Valiant75}.
One way to multiply two Boolean matrices is to treat them as integer matrices, and apply a fast matrix multiplication algorithm over the integers. 
Matrix multiplication over fields can be computed in ``truly subcubic time'', i.e., computing the product of two $n\times n$ matrices can be done in $O(n^{3-\epsilon})$ additions and multiplication over the field. 
For example, the latest generation of such algorithms  run in $O(n^{2.373})$ operations~\cite{VassilevskaW12},~\cite{LeG14}. 
These algorithms are ``algebraic'', as they rely on the structure of the field, and in general the ring structure of matrices over the field. 

There is a different group of BMM algorithms, often called ``combinatorial'' algorithms.
They usually reduce the redundancy in computation by exploiting some combinatorial structure in the Boolean matrices. 
The ``Four Russians'' algorithm by Arlazarov, Dinic, Kronrod, and Faradzhev~\cite{ADKF70} is the most well-known combinatorial algorithm for BMM. On the RAM model with word size $w=\Theta(\log n)$, ``Four Russians'' algorithm can be implemented in $O(n^3/\log^2 n)$ time. About 40 years later, this result was improved by Bansal and Williams. 
In their FOCS'09~\cite{BW09} paper, they presented an $O(n^3(\log\log n)^2/\log^{9/4} n)$ time combinatorial algorithm for Boolean matrix multiplication, using the weak regularity lemma for graphs. Recently, Chan presented an $O(n^3\left(\log\log n\right)^3/\log^3 n)$ time algorithm in his SODA'15 paper~\cite{Chan15}, improving the running time even further. 

Although these combinatorial algorithms have worse running times than the algebraic ones, they generally have some nice properties. 
Combinatorial algorithms usually can be generalized in ways that the algebraic ones cannot be. For example, Chan's algorithm partly extends an idea of divide-and-conquer in an algorithm for the offline dominance range reporting problem by Impagliazzo, Lovett, Paturi, and
Schneider~\cite{ILPS14}; the algebraic structure of dominance reporting is completely different from BMM's. Moreover, in practice, these combinatorial algorithms are usually fast and easy to implement, while in contrast, most theoretically fast matrix multiplication algorithms are impractical to implement. Finding a matrix multiplication algorithm that is both ``good'' in theory and practice is still an important open goal of the area. 

In this paper, we generalize the ideas of Impagliazzo et al.~\cite{ILPS14} and Chan~\cite{Chan15}, to present a faster combinatorial algorithm for triangle detection: \emph{given an $n$-node graph, does it contain a triangle?}

\begin{theorem}\label{TriDectAlgo}
Given a tripartite graph $G$ on $n$ vertices, we can detect if there is a triangle in $G$ using a combinatorial algorithm in $\hat{O}\left(n^3/\log^4 n\right)$ time on a word RAM with word size $w\geq \Omega(\log n)$. \footnote{We use $\hat{O}(f(n))$ to suppress $\textrm{poly}(\log\log f(n))$ factors in the running time.}
\end{theorem}

Vassilevska Williams and Williams~\cite{VW10} proved that triangle detection and Boolean matrix multiplication are ``subcubic equivalent'' in the following sense: if there is a $O(n^3/g(n))$ time algorithm for triangle detection on $n$-node graphs, then we can use it to solve BMM on $n\times n$ matrices in $O(n^3/g(n^{1/3}))$ time. Together with Theorem~\ref{TriDectAlgo}, this gives a fast combinatorial algorithm for Boolean matrix multiplication. 

\begin{theorem}
There is a combinatorial algorithm to multiply two $n\times n$ Boolean matrices in $\hat{O}\left(n^3/\log^4 n\right)$ time. 
\end{theorem}

Moreover, we generalize the algorithm, and propose a general framework for solving triangle detection combinatorially.

\begin{definition}
The \emph{large subgraph triangle detection problem} with parameters $\alpha$, $\beta$, and $\gamma$ is: given a tripartite graph G = $(A\cup B\cup C, E)$, output a pair $(G',b)$, where 
\begin{itemize} 
\item $G'$ is a subgraph with at least an $\alpha$-fraction of vertices from $A$, $\beta$-fraction of vertices from $B$, and $\gamma$-fraction of vertices from $C$, and 
\item $b = 1$ if $G'$ is triangle-free, and $b = 0$ if $G'$ contains a triangle. 
\end{itemize}
\end{definition}

That is, the large subgraph triangle detection problem is to identify a large subgraph $G'$ of $G$ for which we can conclude whether $G'$ is triangle-free or not. This problem is interesting when we can solve it quickly -- faster than what is known for standard triangle detection. Additionally, we can show that fast algorithms for large subgraph triangle detection imply fast algorithms for triangle detection in general:

\begin{theorem}
Let $n$ be an integer, $0 < \alpha, \beta, \gamma, c\leq 1$, and $G$ be any tripartite graph on vertex sets $A, B, C$ with at least $\sqrt{n}$ vertices in each part. If there is an algorithm $\mathbf{L}$ for large subgraph triangle detection on every such $G$ that runs in $O(c\alpha\beta\gamma |A||B||C|)$ time, then we can solve triangle detection on $n$-node graphs in $O(cn^3+n^{3-\epsilon/2})$ time for any $\epsilon>0$ such that $\alpha\beta>10\epsilon(1+\log\frac{1}{\gamma})$ and $\alpha>10\epsilon(1+\log\frac{1}{\beta})$. 
\end{theorem}

That is, to derive an efficient algorithm for triangle finding, it is sufficient to find and solve an ``easy part'' of the input. This opens a new direction for attacking this problem. 

\vspace{10pt}

{\noindent\bf Related work.} \quad In the work by Bansal and Williams~\cite{BW09}, they used the weak regularity lemma of Frieze and Kannan~\cite{FK99} to discover and exploit small substructures in the graph. Generally speaking, a regularity lemma partitions the vertex set of a graph into disjoint sets, so that the edge distribution between any two sets is ``close to random.'' Bansal and Williams enumerate every triple of sets in the partition: if the subgraph induced by the triple is sparse, finding a triangle in this triple is easy. Otherwise, since the induced subgraph is dense and ``close to random'', the regularity lemma guarantees that it is impossible to check many pairs of vertices \emph{without} finding an edge between them. Integrating the method of Four Russians with the above approach yields an $\hat{O}(n^3/\log^{9/4}n)$ time algorithm for triangle detection. 

Chan~\cite{Chan15} used a very different approach for triangle detection which we now outline briefly\footnote{To keep consistency, the following description will be in the language of triangle finding, although his paper originally presented the algorithm in the language of Boolean matrix multiplication.}. Consider a tripartite graph on vertex sets $(A,B,C)$ such that $|A|\leq \textrm{polylog}(|B|+|C|)$ (if this is not the case, partition the set $A$ into $\textrm{polylog}(|B|+|C|)$-size subsets, and solve them independently). If the edge set between $A$ and $B$ is sparse, triangle detection is easy. Otherwise, there is some node $v\in A$ with many neighbors in $B$. Then the algorithm manually checks every pair of neighbors of $v$ and does two recursive calls. One is on $A\setminus\{v\}$, $B$ and the non-neighbors of $v$ in $C$, and the other is on $A\setminus\{v\}$, non-neighbors of $v$ in $B$ and neighbors in $C$. On one hand, this recursive procedure never puts any pairs of neighbors of $v$ in the branch. This guarantees that the algorithm only manually checks every pair in $B\times C$ at most once. On the other hand, the procedure may copy the set of non-neighbors in $B$ when doing a recursive call, which increases the total input size. However, since we have a lower bound on the degree of vertex $v$ to $B$, the procedure does not copy too many vertices each time. A careful analysis shows that the overhead of the recursion is actually rather tiny.

In Section~\ref{secTriFind}, we show how to extend the idea of divide-and-conquer in Chan's algorithm to get an even faster algorithm for triangle detection (and hence for Boolean matrix multiplication). We give a more intuitive proof of the recursion involving less calculation. In Section~\ref{secGeneral}, we propose a general framework for solving triangle detection, as a starting point for future work in this area.

\section{Triangle Detection}\label{secTriFind}

{\noindent\bf Preliminaries and notations.}\quad We shall use the fact that the triangle detection problem in general undirected graphs is time-equivalent to the problem restricted to tripartite graphs, up to a constant factor. (The proof is straightforward.) Henceforth, we will assume the input graph is tripartite, and the tri-partition of its vertices is given to us. 

For a graph $G=(V,E)$, a vertex $v\in V$ and a subset of vertices $S\subseteq V$, we denote $d(v, S)=\left|E\cap(\{v\}\times S)\right|$, which we call the \emph{degree of $v$ to $S$}. 

\vspace{10pt}

{\noindent\bf Main results.}\quad In what follows, we present an $\hat{O}(n^3/\log^4 n)$ time combinatorial algorithm for triangle detection. 

Suppose we are given a tripartite graph $G$ on vertex sets $A, B, C$. One (naive) approach to detect if there is a triangle in the graph is: for vertex $v\in A$, and all pairs $(u,w)$ of $v$'s neighbors, check if there is an edge between $u$ and $w$. 
The amount of work we do for vertex $v$ is proportional to the number of edges between $v$ and $B$ (the degree of $v$ to $B$) times the number of edges between $v$ and $C$ (the degree of $v$ to $C$). 
This approach is efficient whenever the product of these two degrees is low on average. However, if this is not the case, there must be a vertex $v\in A$ with a large product of degrees. If the enumeration reports no edge between any pair of $v$'s neighbors, we know there has to be a large ``non-edge area'' between $B$ and $C$, i.e. between $v$'s neighbors. 
In the rest of the algorithm, there is no need to look again at any pair of vertices in that area. We implement this idea by recursion: find disjoint subsets of $B\times C$ which together cover all pairs outside the non-edge area, and recurse on them. We show this recursion is actually efficient. 

Before stating and proving the efficiency of our main algorithm, we first show that the naive approach proposed above is ``Four-Russianizable'' as expected. That is, we can apply the Method of Four Russians to speed up the sparse case by a factor of roughly $\log^2 n$. 

The following algorithm is a generalization of an algorithm of Bansal and Williams~\cite{BW09}. 
\begin{lemma}\label{SparseAlgo}
Let $G=(V,E)$ be tripartite on vertex sets $A,B,C$ of sizes $k,m,n$ respectively. 
If $d(v_i,B)d(v_i,C)\leq \frac{nm}{\Delta^2}$ holds for some $\Delta$ and all $v_i\in A$ simultaneously, then we can detect if there is a triangle in $G$ combinatorially in  $O(mn\Delta^{6\Delta}+\frac{kmn}{\Delta^4}+k(m+n))$ time, on word RAM with word size $w\geq \Omega(\Delta\log\Delta+\log kmn)$. 
\end{lemma}
\begin{proof}
First, partition the vertices in $B$(resp. $C$) into groups $\{B_i\}$(resp. $\{C_i\}$) of sizes $\Delta^3$ arbitrarily, e.g. put $(i\Delta^3+1)$-th to $(i+1)\Delta^3$-th vertex of $B$(resp. $C$) in $B_i$(resp. $C_i$). Let $\mathcal{S}_B=\{S:|S|\leq\Delta,S\subseteq B_i\textrm{ for some $B_i$}\}$, $\mathcal{S}_C=\{S:|S|\leq\Delta,S\subseteq C_i\textrm{ for some $C_i$}\}$ be collections of subsets within the same group of $B$ or $C$ with at most $\Delta$ vertices. For every $S\in \mathcal{S}_B,S'\in\mathcal{S}_C$, we determine if there is at least one edge between them, and store all the results in a lookup table. This preprocessing takes 
\[
	O\left(\frac{mn}{\Delta^6}{\Delta^3\choose \Delta}^2\cdot \Delta^2\cdot \Delta^2\right)\leq O(mn\Delta^{6\Delta})
\]
time. Note that we can index a subset using $O(\Delta\log\Delta+\log \max\{m,n\})=O(w)$ bits. This table can be stored in the memory so that one table lookup takes constant time. 

With the help of this table, we can check if there is a triangle in $G$ efficiently. We go over all vertices $v_i\in A$, and partition its neighborhood into a minimum number of sets in $\mathcal{S}_B,\mathcal{S}_C$. That is, for every group of vertices, we arbitrarily partition $v_i$'s neighborhood in this group into sets of size exactly $\Delta$ and (possibly) one more set of size at most $\Delta$. 
 This generates at most $\frac{m}{\Delta^3}+\frac{d(v_i,B)}{\Delta}$ sets from $\mathcal{S}_B$ and at most $\frac{n}{\Delta^3}+\frac{d(v_i,C)}{\Delta}$ sets from $\mathcal{S}_C$. Using the lookup table, we can detect if there is an edge between any pair of sets from $\mathcal{S}_B$ and $\mathcal{S}_C$ in constant time. Going over all $v_i\in A$ takes
\[
	\begin{aligned}
		&O\left(\sum_{i=1}^k\left(\frac{m}{\Delta^3}+\frac{d(v_i,B)}{\Delta}\right)\left(\frac{n}{\Delta^3}+\frac{d(v_i,C)}{\Delta}\right)+k(m+n)\right) \\
		&\leq O\left(\frac{kmn}{\Delta^6}+\sum_{i=1}^k\frac{m}{\Delta^3}\cdot\frac{n}{\Delta}+\sum_{i=1}^k\frac{m}{\Delta}\cdot\frac{n}{\Delta^3}+\sum_{i=1}^k\frac{mn}{\Delta^4}+k(m+n)\right) \\
		&\leq O\left(\frac{kmn}{\Delta^4}+k(m+n)\right)
	\end{aligned}
\]
time. The total running time is at most $O\left(mn\Delta^{6\Delta}+\frac{kmn}{\Delta^4}+k(m+n)\right)$ as we stated. 

\qed
\end{proof}

Using this algorithm for the sparse case as a subroutine, we give a fast combinatorial algorithm for triangle detection. 

\begin{restate}
Given a tripartite graph $G$ on $n$ vertices, we can detect if there is a triangle in $G$ using a combinatorial algorithm in $\hat{O}\left(n^3/\log^4 n\right)$ time on a word RAM with word size $w\geq \Omega(\log n)$. 
\end{restate}
\begin{proof}
Set parameter $\Delta=\frac{\log n}{100\left(\log\log n\right)^2}$. It will remain fixed as we do the recursion, even if the instance size shrinks. The following algorithm detects if there is a triangle in a tripartite graph with vertex sets $A,B,C$: 

\begin{framed}
\begin{changemargin}{0.6cm}{0cm} 
\begin{enumerate}
	\item[Step 0:]
		If $|B|<\Delta^6$ or $|C|<\Delta^6$, solve the instance by exhaustive search and return the answer. 
	\item[Step 1:]
		If for all vertices $v_i\in A$, $d(v_i,B)d(v_i,C)\leq \frac{|B|\cdot |C|}{\Delta^2}$, we solve the instance by the algorithm in Lemma~\ref{SparseAlgo}.
	\item[Step 2:]
		Otherwise, find a vertex that violates the condition. Without loss of generality, assume $v_1$ does, $d(v_1,B)d(v_1,C)>\frac{|B|\cdot|C|}{\Delta^2}$. 
	\item[Step 3:]
		Let $B_1$(resp. $C_1$) be $v_1$'s neighborhood in $B$(resp. $C$). \\		
		If $\frac{|B_1|}{|B|}>\frac{|C_1|}{|C|}$, then recurse on $(A\setminus\{v_1\},B,C\setminus C_1)$ and $(A\setminus\{v_1\},B\setminus B_1,C_1)$, 		
		else recurse on $(A\setminus\{v_1\},B\setminus B_1,C)$ and $(A\setminus\{v_1\},B_1,C\setminus C_1)$. \\		
		Return YES if either of the two recursions returned YES. 
	\item[Step 4:]
		Check all pairs of vertices in $B_1\times C_1$ for an edge. \\
		Return YES if there is an edge, NO otherwise. 
\end{enumerate}
\end{changemargin}
\end{framed}

{\noindent \bf Analysis of the running time.}

In each node of the recursion tree, only some of the following four subprocedures are executed:
\begin{enumerate}
	\item
		If either $|B|$ or $|C|$ is small, we do exhaustive search, which takes $O(|A||B||C|)$ time. 
	\item
		We spend $O\left(|A|(|B|+|C|)\right)$ time to check whether there is a high degree vertex and generate the inputs for two recursive calls. 
	\item
		If $A$ has no high degree nodes, we invoke the algorithm in Lemma~\ref{SparseAlgo} which takes 
		\[
		\begin{aligned}
		&O\left(|B| |C|\Delta^{6\Delta}+\frac{|A||B||C|}{\Delta^4}+|A|(|B|+|C|)\right) \\ 
		&=\hat{O}\left(|B| |C|n^{o(1)}+|A||B||C|/\log^4 n+|A|(|B|+|C|)\right)
		\end{aligned}
		\]
		time. 
	\item
		For every pair of neighbors of $v_1$ in $B$ and $C$, we check if they have an edge. This step takes $O(|B_1||C_1|)$ time. 
\end{enumerate}

To analyse the total running time, we are going to bound the time we spend on the small-graph case (Subprocedure 1) and the time we spend on the large-graph case (Subprocedure 2,3,4) in the entire execution of the algorithm (the whole recursion tree) separately, and sum up these two cases. 

For the case when $|B|\geq \Delta^6$ and $|C|\geq \Delta^6$, Subprocedure 2 is cheap compared to the other steps, taking time $O(|A|(|B|+|C|))\leq \hat{O}\left(n|B||C|/\log^6 n\right)$.
In this case, we mentally charge all the running time to the pairs of vertices in $B\times C$, then sum up over all pairs the cost they need to pay. If $A$ has no high degree vertices, we will run Subprocedure 2 and 3, which takes at most $\hat{O}(|B| |C|n^{o(1)}+|A||B||C|/\log^4 n+|A|(|B|+|C|))\leq \hat{O}\left(n|B||C|/\log^4 n\right)$
time. We charge this running time to all pairs of vertices in $B\times C$ evenly; that is, every pair of vertices gets charged $\hat{O}\left(n/\log^4 n\right)$. If $A$ has a high degree vertex, we will run Subprocedure 2 and 4, which takes at most
\[
	\begin{aligned}
	&O\left(|A|(|B|+|C|)+|B_1||C_1|\right) \\
	&\leq \hat{O}\left(\Delta^2|B_1||C_1|n/\log^6 n+|B_1||C_1|\right) \\
	&\leq \hat{O}\left(n|B_1||C_1|/\log^4 n\right)
	\end{aligned}
\]
time. We charge this running time to the pairs in $B_1\times C_1$ evenly, so every pair gets charged $\hat{O}\left(n/\log^4 n\right)$. 

\begin{claim}
Every pair of vertices is charged at most once, over the entire execution of the algorithm.
\end{claim}

The proof of this claim follows from inspection of the algorithm: the above argument only charges pairs of vertices that are not going into the same recursive branch.

There are $n^2$ pairs at the very beginning. Every pair gets charged at most $\hat{O}\left(n/\log^4 n\right)$. Therefore the running time for the large-graph case is at most $\hat{O}\left(n^3/\log^4 n\right)$.

Next we bound the total running time of Subprocedure 1. This running time is proportional to the number of triples we enumerated in Step 0. Let $T(S)$ be the maximum possible of this number of triples, if we start our recursion from vertex sets $A,B,C$ with $|B||C|\leq S$. 

\begin{claim}\label{claimSmallGraph}
	$T(S)\leq nS$. Moreover, for $S>n\Delta^6$, then 
	\[T(S)\leq \max_{t>1/\Delta,t'>1/\Delta^2}\{T((1-t')S)+T\left(t'(1-t)S\right)\}\]
\end{claim}
\begin{proof}
	$T(S)\leq nS$ is trivial, since we never enumerate any triple more than once. If $S>n\Delta^6$, we have $|B|,|C|>\Delta^6$. That is, we must be starting the recursion from a large-graph case. There is nothing to prove if there is no high degree vertex in $A$, as we will do no enumeration in Step 0. Otherwise, let $t=\max\left\{\frac{|B_1|}{|B|},\frac{|C_1|}{|C|}\right\},t'=\min\left\{\frac{|B_1|}{|B|},\frac{|C_1|}{|C|}\right\}$, then by the algorithm, we have $tt'>1/\Delta^2$, in particular, $t>1/\Delta,t'>1/\Delta^2$. In the two recursive calls, we have $|B||C|$ values $(1-t')S$ and $t'(1-t)S$ respectively. By the definition of $T$, we will enumerate at most $T\left((1-t')S\right)+T\left(t'(1-t)S\right)$ triples. This proves the claim. 
	
\qed
\end{proof}

We can upper bound $T(S)$ using this recurrence. Consider the recursion tree $\mathcal{R}$ for $T(S)$. The root has value $S$. Its left child has value $(1-t')S$ and right child has value $t'(1-t)S$, where $t$ and $t'$ maximize $T\left((1-t')S\right)+T\left(t'(1-t)S\right)$. We recursively construct the tree for the left child and right child. 
For a node with value $x$, we always put $(1-t')x$ in its left child and $t'(1-t)x$ in its right child, for $x$'s optimal parameter $t'$ and $t$. 
We expand the tree from nodes with value at least $n^{1.5}(>n\Delta^6)$ recursively. Therefore, we will get a tree with leaf values at most $n^{1.5}$. This tree demonstrates how $T(S)$ is computed according to the recurrence, before we reach $n^{1.5}$. The sum of values of all leaves multiplied by $n$ is an upper bound for $T(S)$, since $T(x)\leq nx$. 

We calculate this sum in two cases. For every leaf, there is a unique path from the root to it, in which we follow the left child in some steps, follow the right child in the rest. Consider all leaves such that the unique path from the root to it takes at most $10\Delta\log\Delta (\approx \frac{\log n}{10\log\log n})$ right-child-moves. Since $t'>1/\Delta^2$, the depth of the tree is at most $\Delta^2\log n$. Therefore, there are at most 
\[
	{\Delta^2\log n \choose 10\Delta\log\Delta}\cdot 10\Delta\log \Delta\leq 2^{3\Delta\log^2\Delta}\leq n^{0.4}
\]
such leaves. By construction, each leaf has value at most $n^{1.5}$, so the sum over all leaves is at most $n^{1.9}$. 

For those leaves such that the path from the root takes more than $10\Delta\log\Delta$ right-child-moves, consider a tree $\mathcal{R}'$ with \emph{identical structure} as $\mathcal{R}$, but we set the values of nodes in a different way. We set all $t$s to $0$ but leave all $t'$s unchanged, then calculate the corresponding values. That is, the root still has value $S$, its left child still has value $(1-t')S$, and its right child now has value $t'S$.
For any node with value $x$ in $\mathcal{R}'$, its left child has value $(1-t')x$ and right child has value $t'x$, for the same ratio $t'$ as the corresponding node in $\mathcal{R}$, although the value $x$ may have changed. Also, leaves now may have values greater than $n^{1.5}$ by the construction of $\mathcal{R}'$. The tree $\mathcal{R}'$ has the following two properties:
\begin{enumerate}
	\item
		The sum of values in all leaves is exactly $S$. 
	\item
		For any node such that the path from root to it takes at least $k$ right-child-moves, its value is at least $1/(1-1/\Delta)^k$ times the value of the corresponding node in $\mathcal{R}$. 
\end{enumerate}

The first property can be proved by induction, since the sum of values of two children is exactly the value of the parent. For the second property, since we require $t>1/\Delta$ in $\mathcal{R}$, and keep $t'$ unchanged, set $t$ to $0$, we will gain a factor of $1/(1-t)>1/(1-1/\Delta)$ for every right-child-move. 

By property 2 above, for all leaves that the path from root takes more than $10\Delta\log\Delta$ right-child-moves, their values in $\mathcal{R}'$ must be at least $1/(1-1/\Delta)^{10\Delta\log\Delta}\geq \Delta^{10}$ times the corresponding values in $\mathcal{R}$. However, the sum of these values in $\mathcal{R}'$ is at most $S$ by property 1. Therefore, the sum of values in $\mathcal{R}$ is at most $S/\Delta^{10}$. Summing these two cases up, we prove that $T(S)\leq nS/\Delta^{10}+n^{2.9}$. In particular, $T(n^2)\leq n^3/\Delta^{10}+n^{2.9}=\hat{O}(n^3/\log^{10} n)$. 

Finally, we sum up the small-graph and large-graph cases, proving that the algorithm runs in $\hat{O}(n^3/\log^4 n)$ time. 

\qed
\end{proof}

{\bf Remark.}\quad In Lemma~\ref{SparseAlgo}, we can preprocess for all subsets of size at most $O(\log k/\log\log k)$ instead of $O(\Delta)$. This improves the running time of Theorem~\ref{TriDectAlgo} from $O\left(\frac{n^3\left(\log\log n\right)^8}{\log^4 n}\right)$ to $O\left(\frac{n^3\left(\log\log n\right)^6}{\log^4 n}\right)$. 

Combining the above algorithm with the reduction by Vassilevska Williams and Williams~\cite{VW10}, we get an efficient combinatorial algorithm for BMM. 

\begin{theorem}[Vassilevska Williams and Williams'10]
For any constant $c$, if we can solve triangle detection on $n$-node graphs in $O(n^3/\log^c n)$ time, we can also solve Boolean matrix multiplication on $n\times n$ matrices in the same running time. 
\end{theorem}

\begin{restate}
There is a combinatorial algorithm to multiply two $n\times n$ Boolean matrices in $\hat{O}\left(n^3/\log^4 n\right)$ time. 
\end{restate}

\section{A General Approach}\label{secGeneral}

In this section, we propose a more general approach for triangle finding, which may lead to an even faster combinatorial algorithm. 

In the algorithm presented in Section~\ref{secTriFind}, finding a high degree vertex $v\in A$ and its neighborhood $B_1,C_1$ can be viewed as finding a large easy part of the input. That is, for the subgraph induced by vertices $A,B_1,C_1$, there is a 2-path for every pair in $B_1\times C_1$.
Thus, we only have to spend $O(|B_1||C_1|)$ time to determine if there is a triangle in it, which is $O(1/|A|)$ time on average for every triple of vertices.
After solving this part of the input, we do two recursive calls which together exactly cover the rest of the triples. 
The high-degree of $v$ guarantees that the easy part we find each time cannot be too small. We will have saved enough time before reaching the case where $B$ or $C$ is close to constant size (in which we basically have no way to beat the exhaustive search). However, if all vertices have low degree, then that instance itself is easy (via Lemma~\ref{SparseAlgo}). 
The following theorem generalizes this idea of reducing the triangle detection problem to finding a large subgraph on which triangle detection is easy. 

\begin{restate}
Let $n$ be an integer, $0 < \alpha, \beta, \gamma, c\leq 1$, and $G$ be any tripartite graph on vertex sets $A, B, C$ with at least $\sqrt{n}$ vertices in each part. If there is an algorithm $\mathbf{L}$ for large subgraph triangle detection on every such $G$ that runs in $O(c\alpha\beta\gamma |A||B||C|)$ time, then we can solve triangle detection on $n$-node graphs in $O(cn^3+n^{3-\epsilon/2})$ time for any $\epsilon>0$ such that $\alpha\beta>10\epsilon(1+\log\frac{1}{\gamma})$ and $\alpha>10\epsilon(1+\log\frac{1}{\beta})$. 
\end{restate}

\begin{proof}

First note that if $\epsilon<\frac{1}{\log n}$, the theorem is trivial, since $n^{3-\epsilon/2}=\Theta(n^3)$. In the rest of the proof, we will assume $\epsilon\geq \frac{1}{\log n}$, and thus $\alpha\beta>\frac{10}{\log n}$. Also, we may assume that $G'$ has sizes exactly $\alpha|A|,\beta|B|,\gamma|C|$ when it is triangle-free, since otherwise we can simply drop extra vertices from the sets. 

The following divide-and-conquer algorithm detects if there is a triangle among vertices $A\cup B\cup C$:

\begin{framed}
\begin{changemargin}{0.6cm}{0cm}
\begin{enumerate}
	\item[Step 0.]
		If $|A||B||C|<n^{2.5}$, do exhaustive search on all triples and return the answer.
	\item[Step 1.]
		Run $\mathbf{L}$ on $A\cup B\cup C$. \\
		Let $G'$ on $A'\cup B'\cup C'$ be the subgraph $\mathbf{L}$ outputs. \\
		Return YES if $G'$ contains a triangle.
	\item[Step 2.]
		Recurse on vertex sets $(A,B,C\setminus C')$, $(A,B\setminus B',C')$, $(A\setminus A',B',C')$. \\
		Return whether any of the three recursive calls returned YES. 
\end{enumerate}
\end{changemargin}
\end{framed}

Its correctness is straightforward. We will prove that the algorithm is efficient. 

For the large-graph case, a similar ``charging argument'' works here as in the proof of Theorem~\ref{TriDectAlgo}. Note that the input $A\cup B\cup C$ we fed to $\mathbf{L}$ in Step 1 always has at least $\sqrt{n}$ vertices in each part. By our assumption, Step 1 will run in $O(c\alpha\beta\gamma |A||B||C|)\leq O(c|A'||B'||C'|)$ time. In Step 2, the algorithm generates the input for recursive calls, which takes only $O(|A|+|B|+|C|)$ time\footnote{It is dominated by the running time of Step 1. Since to determine if $G'$ has a triangle, $\mathbf{L}$ takes at least $\Omega(|A'||B'|)$ time. But we have $|A'||B'|\geq\alpha\beta|A||B|\geq 10n^{1.5}/\log n\geq \Omega(n)$. }. We charge the running time of Step 1 and Step 2 to the triples in $A'\times B'\times C'$. On average, each triple is charged $O(c)$ time. Same as the proof of Theorem~\ref{TriDectAlgo}, every triple in $A'\times B'\times C'$ will not go into the same recursive branch together. Therefore, every triple is charged at most once in this argument. There are $n^3$ triples in total. They are charged at most $O(cn^3)$ time. This proves that Step 1 and Step 2 take at most $O(cn^3)$ time in the entire algorithm.

For the small-graph case, the time we spent on Step 0 is proportional to the number of triples we enumerated. Let $T(S)$ be the maximum possible value of this number, if we start our recursion with $|A||B||C|=S$. On one hand, we have $T(S)\leq S$, as every triple will be manually checked at most once. Moreover, for $S\geq n^{2.5}$, by the way that the algorithm does recursive calls, we have $T(S)\leq T\left(\left(1-\gamma\right)S\right)+T\left(\gamma(1-\beta)S\right)+T\left(\beta\gamma(1-\alpha)S\right)$.

We are going to prove $T(S)\leq n^{2.5}\left(\frac{S}{n^{2.5}}\right)^{1-\epsilon}$ by induction on $S$. For $S\leq n^{2.5}$, the new upper bound automatically holds, since $T(S)\leq S\leq n^{2.5}\left(\frac{S}{n^{2.5}}\right)^{1-\epsilon}$. Otherwise, by the recurrence and induction hypothesis,
\[
	\begin{aligned}
		T(S)&\leq n^{2.5}\left(\left(\frac{\left(1-\gamma\right)S}{n^{2.5}}\right)^{1-\epsilon}+\left(\frac{\gamma(1-\beta)S}{n^{2.5}}\right)^{1-\epsilon}+\left(\frac{\beta\gamma(1-\alpha)S}{n^{2.5}}\right)^{1-\epsilon}\right) \\
		&\leq n^{2.5}\left(\frac{S}{n^{2.5}}\right)^{1-\epsilon}\left(\left({1-\gamma}\right)^{1-\epsilon}+\left({\gamma(1-\beta)}\right)^{1-\epsilon}+\left({\beta\gamma(1-\alpha)}\right)^{1-\epsilon}\right). \\
	\end{aligned}
\]

The value of $\left({1-\gamma}\right)^{1-\epsilon}+\left({\gamma(1-\beta)}\right)^{1-\epsilon}+\left({\beta\gamma(1-\alpha)}\right)^{1-\epsilon}$ is always less than 1: 

\allowdisplaybreaks[1]

\begin{align*}
	&\left({1-\gamma}\right)^{1-\epsilon}+\left({\gamma(1-\beta)}\right)^{1-\epsilon}+\left({\beta\gamma(1-\alpha)}\right)^{1-\epsilon} \\
	&\leq \left(1-(1-\epsilon)\gamma\right)+\gamma^{1-\epsilon}\left(1-(1-\epsilon)\beta\right)+(\beta\gamma)^{1-\epsilon}\left(1-(1-\epsilon)\alpha\right) \tag{by $(1-x)^c\leq 1-cx$ when $0<c,x<1$} \\
	&= 1+\gamma\left(-(1-\epsilon)+\gamma^{-\epsilon}\left(1+\beta\left(-(1-\epsilon)+\beta^{-\epsilon}\left(1-(1-\epsilon)\alpha\right)\right)\right)\right) \\
	&\leq 1+\gamma\left(-(1-\epsilon)+\gamma^{-\epsilon}\left(1+\beta\left(-(1-\epsilon)+(1+2\epsilon\log\frac{1}{\beta})(1-\frac{2\alpha}{3})\right)\right)\right) \tag{by $\epsilon\log\frac{1}{\beta}<1/10$ and $\epsilon<1/10$}\\
	&\leq 1+\gamma\left(-(1-\epsilon)+\gamma^{-\epsilon}\left(1+\beta\left(\epsilon+2\epsilon\log\frac{1}{\beta}-\frac{2\alpha}{3}\right)\right)\right) \\
	&\leq 1+\gamma\left(-(1-\epsilon)+(1+2\epsilon\log\frac{1}{\gamma})(1-\frac{\alpha\beta}{3})\right) \tag{by $\epsilon\log\frac{1}{\gamma}<1/10$ and $\epsilon(1+\log\frac{1}{\beta})<\alpha/10$}\\
	&\leq 1+\gamma\left(\epsilon+2\epsilon\log\frac{1}{\gamma}-\frac{\alpha\beta}{3}\right) \\
	&\leq 1-\frac{\alpha\beta\gamma}{10}<1 \tag{by $\epsilon(1+\log\frac{1}{\gamma})<\alpha\beta/10$}.\\
\end{align*}

This proves $T(S)\leq n^{2.5}\left(\frac{S}{n^{2.5}}\right)^{1-\epsilon}$, in particular, $T(n^3)\leq n^{3-\epsilon/2}$. Therefore, summing up the two cases, the total running time of the algorithm is at most $O(cn^3+n^{3-\epsilon/2})$. 

\qed
\end{proof}

{\noindent\bf Remark.} Note that we do not have to restrict ourselves to find an $\alpha$-fraction of $A$, $\beta$-fraction of $B$ and $\gamma$-fraction of $C$. As long as we can find one part with $\alpha$-fraction, one with $\beta$-fraction and the third with $\gamma$-fraction, and adjust the inputs for recursive calls correspondingly,
we will be able to solve triangle detection efficiently. In this sense, the algorithm in Section~\ref{secTriFind} has parameters $c=\frac{1}{\Delta^4},\alpha=1,\beta=\frac{1}{\Delta},\gamma=\frac{1}{\Delta^2},\epsilon=\Theta(\frac{1}{\Delta\log\Delta})$, while $G'$ is the subgraph induced by $(A,B_1,C_1)$ or the entire graph if all vertices in $A$ have low degree.

\section{Conclusion}

We have shown how to generalize the idea of divide-and-conquer in Chan's algorithm, and have provided a more intuitive proof of the recursion. 
The way of analysing the ``sublinear'' recurrence in Theorem~\ref{TriDectAlgo}, i.e., our $T(S)$ and its analysis, should be able to extend to other problems. We would like to see more applications of this method in proving the efficiency of other combinatorial algorithms that are based on divide-and-conquer. 

Also, we would hope to have an $O(n^2)$ time algorithm for triangle detection on tripartite graphs with vertex set sizes $n, n$, and $\hat{O}(\log^4 n)$. We call this the ``lopsided'' triangle detection problem, where one side of vertices is very small compared to the others. An argument similar to the reduction by Vassilevska Williams and Williams~\cite{VW10} shows that this would yield an $O(n^2)$ time algorithm for multiplying $n\times n$ and $n\times\hat{O}(\log^4 n)$ Boolean matrices, improving the maximum outer dimension $d$ that $n\times n$ and $n\times d$ Boolean matrices can be multiplied in $O(n^2)$ time, in both the combinatorial and the algebraic world. The current record of $d$ is $\hat{O}(\log^3 n)$ by Chan's algorithm (Chan gave an $O(n^2)$ time algorithm for multiplying $n\times d$ and $d\times n$ matrices, which implies an $O(n^2)$ time triangle finding algorithm), while the record in the algebraic world is merely $O(\log n)$~\cite{BD76}. 

Finally, we provide one type of instance for the lopsided triangle detection problem which seems hard to solve in $O(n^2)$ time with our current techniques: a graph $G$ on vertex sets $A,B,C$ with $|A|=|B|=n$ and $|C|=\hat{O}(\log^4 n)$, with roughly $1/\log n$ fraction of edges between $A$ and $C$, $1/\log n$ fraction of edges between $B$ and $C$, and constant fraction of edges between $A$ and $B$. The main difficulty is that the size of $C$ is too small. If we try to do recursion, its size will reach a constant too soon. Once it becomes of constant size, we basically have no way to save anything from exhaustive search. 

\vspace{10pt}

{\noindent \bf Acknowledgement.} The author would like to thank Ryan Williams for helpful discussions on results and writing of the paper, and the anonymous reviewers for their valuable comments. 

\bibliographystyle{plain}
\bibliography{cbmm}

\end{document}